\documentclass[conference]{IEEEtran}
\IEEEoverridecommandlockouts

\usepackage{color,latexsym,amsfonts,amssymb}
\usepackage{amsmath,cite}
\usepackage{amsmath,amsthm,amsbsy}
\usepackage[mathscr]{euscript}
\usepackage{tikz}
\usepackage{url}
 \usepackage{paralist}

\usepackage{algorithm}
\usepackage{algpseudocode}

\usepackage[compact]{titlesec} 

\usepackage{caption}
    \titlespacing{\section}{0pt}{2ex}{1ex}
    \titlespacing{\subsection}{0pt}{1ex}{0ex}
    \titlespacing{\subsubsection}{0pt}{0.5ex}{0ex}
\newcommand{\red}[1]{\textcolor{black}{#1}}

\newcommand{\E}{\mathbf{E}}

\newcommand{\Prob}{\mathbf{P}}

\newcommand{\R}{\mathbb{R}}

\newcommand{\calP}{\mathcal{P}}

\newcommand{\SINR}{\textnormal{SINR}}

\def\fade{{F}} 

\newcommand{\threshold}{{\tau}} 
\newcommand{\noise}{{W}} 

\newcommand{\GRate}{\mathcal{R}} 
\newcommand{\Rate}{\mathcal{R}} 

\newcommand{\Thpt}{\mathcal{T}} 

\newcommand{\statespace}{{\cal{X}}}

\newcommand{\event}{\psi}

\DeclareMathOperator*{\argmax}{\arg\!\max}

\theoremstyle{plain} 
\newtheorem{Theorem}{Theorem}[section]
\newtheorem{Corollary}[Theorem]{Corollary}

\newtheorem{Proposition}[Theorem]{Proposition}

\theoremstyle{definition}

\theoremstyle{definition}
\newtheorem{definition}{Definition}[section]

\theoremstyle{definition} 
\newtheorem{Example}{Example}[section]

\theoremstyle{remark} 
\newtheorem{Remark}{Remark}[section]

\begin{document}

\title{Adaptive determinantal scheduling with fairness in wireless networks}

\author{H.P. Keeler and B. B{\l}aszczyszyn  
\thanks{{\bf H.P. Keeler} is  with the {\em University of Melbourne}, Melbourne, Australia. {\bf B. B{\l}aszczyszyn} is with {\em Inria/ENS}, Paris, France. }
\vspace{-1ex}
}
\date{\today}
\maketitle

\begin{abstract}
We propose a novel framework for wireless network scheduling with fairness using determinantal (point) processes. Our approach incorporates the repulsive nature of determinantal processes, generalizing traditional Aloha protocols that schedule transmissions independently. We formulate the scheduling problem with an utility function representing fairness. We then recast this formulation as a convex optimization problem over a certain class of determinantal point processes called $L$-ensembles, which are particularly suited for statistical and numerical treatments. These determinantal processes, which have already proven valuable in subset learning, offer an attractive approach to network resource scheduling and allocating. We demonstrate the suitability of determinantal processes for network models based on the signal-to-interference-plus-noise ratio (SINR). Our results highlight the potential of determinantal scheduling coupled with fairness. This work bridges recent advances in machine learning with wireless communications, providing a mathematically elegant and computationally tractable approach to network scheduling.
\end{abstract}

\section{Introduction}
Scheduling and allocating resources demands careful consideration in wireless networks and other complex systems. The  algorithms must efficiently distribute time, power, frequency, and other resources among network users, responding dynamically to changing network conditions and traffic demands. The existing literature spans a wide spectrum of approaches for scheduling and allocating resources, ranging from simple heuristics to sophisticated mathematical models. These techniques seek to optimize network resources with varied aims, such as maximizing network capacity,  throughput, and fairness~\cite{asadi2013survey,xu2021survey,dong2020deep}

At their core, schedulers are algorithms for selecting subsets from a broader set, with the goal of optimizing a certain utility function. When allocating resources by choosing subsets, a common requirement is some degree of repulsion in the subsets, meaning the scheduler should not select resources that are close to each other in some way. This is particularly the case when working with the signal-to-interference-plus-noise ratio (SINR). Besides transmitter-receiver distances, the main influence is the clustering of network nodes reducing SINR values. This effect strongly hints at choosing a scheduling (or subset selecting) method that has inherent repulsion.

\subsection{From fermions to scalable learning models}
Determinantal point processes, originally proposed to model repulsive particles known as fermions such as electrons, possess interesting and convenient mathematical characteristics~\cite{hough2006determinantal}. 
Determinantal processes are amenable to statistical inference methods, while they can also be easily and \emph{exactly} simulated on computers using an algorithm that takes a matrix eigen-decomposition~\cite{lavancier2015determinantal}. 


Over the last decade the machine learning research community has identified determinantal point processes as particularly promising tools for selecting complex subsets. Starting with the pioneering efforts by Kulesza and Taskar~\cite{kulesza2012determinantal}, researchers have successfully leveraged determinantal processes to develop sophisticated machine learning frameworks~\cite{gartrell2019learning}, demonstrating their potential for automated subset selection and generative models in general. 

Determinantal models are parametric and relatively scalable, whereas other models commonly suffer from an explosion in computational complexity.  Resource allocation as subset selection is typically challenging, often falling into the dreaded NP-complete complexity class, recalling that if there are $n$ elements in a set, there are $2^n$ subsets. But determinantal point processes are well-suited for handling subset problems, and we will see that they naturally extend a classic approach for scheduling.

\subsection{From spatial Aloha to proportional fairness} 
\label{ss.spatial_aloha} 
For wireless networks, selecting subsets of nodes to communicate is usually dictated by \emph{medium access control} (MAC) protocols. A classic random one is the (spatial) Aloha scheme in which network transmitters independently access the network with some probability $p$, where $p$ is a fixed constant sometimes called the \emph{medium access probability}. Proposed in the 1970s~\cite{abramson1970aloha}, this scheduling scheme is compatible with Poisson network models, where the transmitters are scattered across the plane $\R^2$ according to a Poisson point process $\Phi=\{X_i\}_i$ with intensity $\lambda>0$. Under the assumption of discrete time, at any time instant the transmitters accessing the network will form another Poisson point process  $\Phi_{\red{p}}=\{X_i\}_i$  with intensity $\lambda p$.



\subsubsection{Constant Aloha in Poisson bi-pole model}
\label{ss.Constant-Aloha}
Baccelli, B{\l}aszczyszyn and M\"uhlethaler~\cite{BBM03_allerton,BBM06IT,JSAC} studied a Poisson network model $\Phi_p$  with an Aloha scheme. They assumed each transmitter is associated to a single receiver located uniformly around the transmitter at a fixed distance $r$, resulting in the so-called \emph{bi-pole} or \emph{bi-polar model}; see the monograph by Baccelli and  B{\l}aszczyszyn~\cite[Chapter 16]{FnT2}. They were interested the value of probability $p$ that maximizes the overall coverage based on SINR values exceeding some threshold $\threshold>0$. To this end, they assumed a receiver at the origin $o$ and framed the problem in terms of the \emph{spatial throughput} defined as
\begin{equation}
p\lambda\Prob^o(\SINR(X_o,o,\Phi_p)> \threshold) \, ,
\end{equation}
where $\SINR(X_o,o,\Phi_p)$ is the SINR (which we define later) at the typical receiver $o$ with signal coming from  its transmitter $X_o$  and  $\Prob^o$ represents the Palm distribution of the typical receiver at the origin.
We can interpret the spatial throughput as the average spatial density of transmitter--receiver pairs communicating with the SINR greater than threshold~$\tau$. The resulting optimization problem requires finding the value of $p$ that maximizes the spatial throughput, namely
\begin{equation}
p^*:= \argmax_{0\leq p \leq 1} \left[ p\lambda\Prob^o(\SINR(\red{X_o,o},\Phi_p)> \threshold) \right] \,.
\end{equation}

For the path loss model, Baccelli, B{\l}aszczyszyn and M\"uhlethaler used the standard (singular) power-law  model
$\ell(x)=(\kappa|x|)^{\beta }$, where the constants $\beta>2$ and $\kappa>0$. They also assumed each wireless signal experienced independent Rayleigh fading by using independent and identically distributed (i.i.d.) exponential random variables. These model assumptions, which now form a well-studied network model, gave results showing that the spatial throughput is maximized for the probability 
$$
p^*=\min\left[1,\frac{\beta\sin(2\pi/\beta)}{2\pi^2}r^{-2}\threshold^{-2/\beta}\right]\,.
$$

\subsubsection{Adaptive, proportionally fair Aloha}
\label{ss.Adaptive-Aloha}
Baccelli, B{\l}aszczyszyn and Singh~\cite{baccelli2014analysis} improved on the above results in two main ways. They examined the case for transmitters being allowed to have different $p$ values depending on the Poisson network configuration $\Phi$. (This ability to have different $p$ value for each transmitter motivates the term \emph{adaptive}.) Furthermore, they used a log function as a utility function to study \emph{proportional fairness} in networks, inspired by the seminal paper by Kelly, Maulloo and Tan~\cite{kelly1998rate}. 

For their network model, Baccelli, Błaszczyszyn and Singh assumed both a power-law path loss function $\ell(x)=(\kappa|x|)^{\beta }$ and independent Rayleigh fading $\fade$.  Their final solution for the Aloha probability is a function, which they refer to as a \emph{policy} dependent on the \red{ geometry of the bi-pole model, and in particular} the transmitter location $x\in\R^2$ \red{ with the  receiver at $o$}. 
\red{This probability function \( p^*(x;\Phi) \) is the solution to the optimization problem  
\begin{align}
&p^{*}(x;\Phi):= \label{e.p*}
\argmax_{0 \leq \red{p(\cdot;\cdot)}\leq 1} \\ &\left[ \E^o  \left(\log[p(\red{X_o;\Phi}) \lambda \Prob^o(\SINR(\red{X_o},o,\red{\Phi_{p(\cdot;\cdot)}})>\threshold| \Phi)  \right) \right]\nonumber
 \,,
\end{align}
where \( x \in \mathbb{R}^2 \) is the transmitter's location, with its receiver at the origin \( o \). The expectation \( \E^o \) is taken under the Palm distribution \( \Prob^o \) of the typical receiver. The conditional probability \( \Prob^o(\cdot \mid \Phi) \) accounts for random fading variables, while the Aloha scheme is governed by the translation-invariant function $p(\cdot;\cdot)$. Specifically, in the bi-pole network \( \Phi_{p(\cdot;\cdot)}\), a transmitter at \( x \) communicating with its receiver at \( y \) accesses the node independently with probability \( p^*(x - y, \Phi - y) \).} 
Remarkably, by invoking the Mass Transport Principle, Baccelli, B{\l}aszczyszyn and Singh~\cite{baccelli2014analysis} offered an explicit and, as argued, unique construction of the optimal probability function $p^*(\cdot;\cdot )$.  The authors also considered a localized version of this function, where transmitters compute their values based solely on the geometry 
within a defined local region, which is formalized as a \emph{stopping set}. This localized calculation accounts for the impact of the network configuration outside this set on the aforementioned Palm distribution $\Prob^o$ of the typical receiver. Finally, the optimization problem~\eqref{e.p*} can be viewed through the lens of the more general theory of optimal stationary markings, as presented by B{\l}aszczyszyn and Hirsch~\cite{BLASZCZYSZYN2021153}.


\section{Related work}

\subsection{Scheduling and fairness}
There is no shortage of proposed methods for scheduling and allocating resources in wireless networks~\cite{xu2021survey}, as well as other areas such as cloud computing~\cite{vinothina2012survey}.
The most relevant research to the current work is that covering adaptive Aloha, which we detailed above with citations in Section~\ref{ss.spatial_aloha}.

There is also an extensive range of research on incorporating fairness into allocating resources, as it is a central problem in computing and communication systems in general. Applications range from wireless networks~\cite{khan2016fairness} to more recent trends in cloud computing~\cite{joe2018harnessing}.  Also see the recent paper by Si Salem, Iosifidis, and Neglia~\cite{si2022enabling} for recent overview of fairness approaches.

There are different types of fairness, such as \emph{alpha} (or $\alpha$) fairness and \emph{max-min} fairness. Here we use \emph{proportional} fairness, which was the subject of a highly influential paper by Kelly, Maulloo and Tan~\cite{kelly1998rate} who cast fairness as an optimization problem,  partly inspiring our current work. 


\subsection{Determinantal (statistical) learning}
Starting with the seminal work by Kulesza and Taskar~\cite{kulesza2012determinantal}, the machine learning research community has steadily examined and developed models based on determinantal processes. This work includes adapting determinantal models to handle positive correlation (or clustering)~\cite{brunel2018learning} and methods to train them~\cite{gartrell2019learning}.

\subsection{Determinantal scheduling}  
In  earlier work~\cite{blaszczyszyn2018determinantal}, we used a determinantal (discrete) point process and introduced a novel point process defined on bounded regions of the plane $\mathbb{R}^2$, which offers model capabilities such as a closed-form expression for the Laplace functional and Palm distribution. We  discussed the potential of using a determinantal scheduler in a wireless network and then training (or fitting) this model on pre-optimized network configurations~\cite[Section VI. C.]{blaszczyszyn2018determinantal}. 

We then proposed in a second paper~\cite{blaszczyszyn2020coverage} a determinantal scheduler based on the SINR of a wireless network. We gave two examples of wireless network models, for which  we derived expressions for the Palm distribution and Laplace functional. Using these tools, we obtained closed-form expressions for the coverage probabilities in the two respective models. Critically, this work laid the groundwork for exploring determinantal scheduling in wireless network by presenting fundamental results for coverage probabilities.

Concurrently, other researchers have begun exploring similar determinantal approaches. Shortly before and independently of our determinantal scheduling work~\cite{blaszczyszyn2020coverage}, Saha and Dhillon~\cite{saha2019machine} applied discrete determinantal point processes to wireless link scheduling, focusing on maximizing the overall network rate. They found optimal network subsets with geometric programming, which they then used to train their parametric determinantal model. In our previous paper~\cite{blaszczyszyn2020coverage}, we discussed how our work diverged from their work~\cite{saha2019machine}. 

Recently, Tu, Saha, and Dhillon~\cite{tu2023determinantal,tu2025determinantal} continued the above line of research, which focuses on training the model, thus complementing our investigations. Tu, Saha, and Dhillon~\cite{tu2023determinantal} started by considering a determinantal scheduler in a drone network model. The determinantal scheduler can be defined with a square matrix $S$, which is typically positive semi-definite. Tu, Saha, and Dhillon~\cite{tu2023determinantal} proposed a $S$ matrix based on the network interference.

Tu, Saha, and Dhillon~\cite{tu2023determinantal} also examined the other part of the determinantal scheduler, which is a vector $q$ capturing the quality of each element in a set.  They developed a parametric quality model based on the SINR of each network node. They then found optimal subsets of network of transmitter-receiver pairs, meaning a bi-pole model, by employing geometric programming, a classical optimization approach. Tu, Saha, and Dhillon later extended this work in a recent preprint~\cite{tu2025determinantal}, which includes another network model. This~\cite{tu2025determinantal} and the aforementioned work~\cite{saha2019machine,tu2023determinantal} examined instantaneous rates. Expressions for coverage probabilities, which rely upon the determinantal properties, were not employed. No type of fairness was considered.

\subsection{Current contributions}
We introduce a general optimization framework based on time-averaged rates. We then focus on a logarithmic utility function designed for proportional fairness. We apply the proportional fairness to a network under a scheduling scheme based on determinantal point processes. For an example, we apply this scheduler to a bi-pole network model. We recast the SINR-based coverage probability into a tractable  kernel for determinantal  process. We achieve this by using the tractable Palm and algebraic properties of determinantal processes.

\section{Random medium access control}
We consider a wireless network model as a finite set of points $\statespace=\{x_1,\ldots,x_n\}$.
By a (random) \emph{medium access control (MAC)} scheduler, we understand a distribution of the {\em random subset} $\Psi\subset\statespace$ of nodes, which are authorized (that is, scheduled) to be simultaneously active (that is, transmitting). In practice, the MAC scheduler $\Psi$ is an algorithm allowing one to sample from the probability distribution of the random set~$\Psi$.

\begin{Example}[Aloha]
The constant and adaptive Aloha schemes described in Section~\ref{ss.Constant-Aloha}
and~\ref{ss.Adaptive-Aloha} are examples of a random MAC scheduler, where the random set $\Psi$ is selected using {\em (independent) Bernoulli} sampling with constant probability $p$ or variable probability $p(x_i)= 
p(x_i;\statespace)$
respectively.
\end{Example}

\subsection{Transmission rate and throughput}
We assume that whenever a subset $\psi\subset\statespace$ of nodes is simultaneously active,  then the {\em transmission rate} $\GRate_i(\psi)\ge0$ is established in the node $x_i\in\statespace$,
with $\GRate_i(\psi)=0$ when $x_i\not\in\psi$. For a given MAC scheduler (distribution of) $\Psi$, we define the {\em throughput} $\Thpt_i$ of the node $x_i\in\statespace$  as the 
average transmission rate achieved in the node~$x_i$
\begin{equation}\label{e.Throughput}
\Thpt_i:=\E[\GRate_i(\Psi)].
\end{equation}
We stress that $\Thpt_i$ corresponds to the time-average rate of transmitter $x_i$ when the scheduler assigns transmissions in a time-stationary way with the marginal distribution of the scheduled transmissions corresponding to $\Psi$. 

We now consider some examples of the  rate function $\GRate_i$.
\begin{Example}
Consider  rates of the general form
\begin{equation}
\GRate_i(\psi)=\bar{\GRate}_i g(|\psi|)\,,     
\end{equation}
where $g(n)\ge0$ is some given function, $|\psi|$ is the cardinality of~$\psi$ and $\bar{\GRate}_i$ are the ``peak rates'' of the nodes. 
    \textbf{Case 1.} $g(n)=1$ corresponds to the rates obtained when the nodes are scheduled separately.
    \textbf{Case 2.}   $g(n)=1/n$ gives the policy known as {\em round robin}.
   \textbf{Case 3.}   $g(n)=(\sum_{k=1}^n1/k)/n$ arises when 
$\GRate_i(\psi)=\bar{\GRate}_i\E[\max_{x_j\in\psi}: F_j]/n$, for $x_i\in\psi$, where $F_i$
are i.i.d. unit-mean exponential random variables.
This last case is the so-called \emph{opportunistic} scheduling with respect to random propagation effects; see the paper by Borst~\cite{borst2005user}.
\end{Example}

\subsection{Utility and fairness}
For a given network configuration $\statespace$, we can optimize the performance of the MAC scheduler $\Psi$ by first considering some utility function $U$ of the throughput, and then maximizing the total utility
\begin{equation}
U_{\Psi}:= \sum_{x_i\in\statespace}U(\Thpt_i).
\end{equation}
To reduce subscripts, we will often write such expressions involving sums or optimums taken over $\statespace$  using just $i$, implying $U_\Psi= \sum_{i}U(\Thpt_i).$
Using this total utility $U_{\Psi}$ we define our optimal scheduler $\Psi$.
\begin{definition}[Scheduler $\Psi$ with optimal $U_{\Psi}$]\label{def.Psiopt}
For a given network configuration $\statespace$, subsets $\psi\subset\statespace$,  the rate functions $\GRate_i(\psi)$, and a concave utility function $U$, we define the optimal scheduler $\Psi$ as the solution of general optimization problem 
\begin{equation}\label{e.Utility-optimization}
\argmax_{\mathcal{L}(\Psi)} U_\Psi,  \end{equation}
where the maximum is taken over all probability distributions 
$\mathcal{L}(\Psi)$ 
of (random) subsets $\Psi\subset\statespace$.
\end{definition}
In the terminology of Kelly, Maulloo and Tan~\cite{kelly1998rate}, we can now  define a \emph{proportionally fair} scheduler $\Psi$.
\begin{definition}[Proportionally fair scheduler $\Psi$]\label{def.profair}
The scheduler $\Psi$ is a {\em proportionally fair}  scheduler when it solves the optimization problem~\eqref{e.Utility-optimization} with the logarithmic utility function $U(t)=\log(t)$. In other words, this scheduler is the solution to the optimization problem 
\begin{equation}
\argmax_{\mathcal{L}(\Psi)} \sum_{i}\log(\Thpt_i)\, . 
\end{equation}
\end{definition}
Note that any scheduler $\Psi$ that systematically excludes any transmitters will have logarithmic utility $U_\Psi(\statespace)=-\infty$. Hence, a natural assumption is that $\Prob(x_i\in\Psi)>0$
for all $x_i\in\statespace$. Furthermore, any Aloha scheme with $p(x_i)>0$ for all nodes $x_i$ will have a finite utility.

\begin{Remark} 
The solution set of the optimization problem~\eqref{e.Utility-optimization} is non-empty and consists of a unique probability distribution
$\mathcal{L}(\Psi)$ that represents both a local and, consequently, a global maximum of~$U_\Psi$.
This is due to the fact that the function $U_\Psi$ is a continuous and convex function that maps from the compact, convex set of probability distributions defined on the power set $2^\statespace=\{\psi: \psi\subset\statespace\}$. For details, consult a text on (convex) optimization methods such as Boyd and Vandenberghe~\cite{boyd2004convex}.
\end{Remark}



\subsection{Scheduler as a determinantal point process}
\label{ss.DetMAC}
We start by defining a determinantal point process on a discrete (finite) state space $\statespace$.
\begin{definition} [Determinantal point process]\label{def.detpp}
For a real, positive semi-definite matrix $K$ indexed by the points of a state space $\statespace$,  having  all its eigenvalues in the interval $[0,1]$, a determinantal point process $\Psi$ on state space $\statespace$
is defined by its finite-dimensional probabilities
\begin{equation}\label{e.probdet}
\Prob(\Psi\supseteq  \event  ) = \det(K_{\psi}),    
\end{equation}
where $\det$ denotes the determinant, and $K_{\event}:=[K]_{x_i,x_j\in {\event}}$ denotes the restriction of $K$ to the  points $\psi\subset\statespace$. 
\end{definition}
Typically, we will interpret this discrete random point process as a random subset of the underlying state space $\statespace$, which we now use to define our  scheduler. 
\begin{definition} [Determinantal determinantal scheduler $\Psi$]  
We define a determinantal MAC scheduler $\Psi$ as a random subset $\Psi\subset\statespace$
whose distribution is a determinantal process on some underlying state space $\statespace$
with some {\em (marginal)  kernel} $K$, as given in Definition~\ref{def.detpp}. 
\end{definition}

\begin{Remark}
The determinantal MAC scheduler is a natural extension of the adaptive Aloha of Section~\ref{ss.Adaptive-Aloha}, where we consider it here for a finite network. We can see this by taking the diagonal matrix kernel 
$K=\mathrm{diag}[(p(x_i))_{x_i\in\statespace}]$,  for any function $p(x)$, which
makes $\Psi$ an (independent) Bernoulli thinning of $\statespace$.
\end{Remark}
A determinantal point process is stochastically \emph{repulsive}, meaning that the presence of a point in a given region reduces the probability of finding another point within the same region.
Intuition suggests that the repulsive nature, so points tend not not cluster together, should lead to better performance (larger utility) of the determinantal scheduler. The  existence of efficient sampling algorithms for (discrete) determinantal point processes is yet another motivation to use them as MAC schedulers.

We recall the throughput $\Thpt_i=\E[\GRate_i(\Psi)]$ and write $\mathcal{L}_K(\Psi)$ to denote all the determinantal probability distributions on the space $\statespace$, which corresponds to all the possible determinantal kernels $K$. We now define the determinantal version of the proportionally fair scheduler defined in Definition~\ref{def.profair}.
\begin{definition}[Proportionally fair,  determinantal scheduler $\Psi$]\label{def.profairdet}
We define a {\em proportionally fair, determinantal MAC scheduler} as any solution of the  optimization problem 
 \begin{equation}\label{e.Det-Utility-optimization}
\argmax_{\mathcal{L}_K(\Psi)}\sum_{i}\log(\Thpt_i)=\argmax_{K}\sum_{i}\log(\Thpt_i)\, ,
\end{equation}
where  the maximum is taken over all marginal determinantal  kernels $K$ of the determinantal process $\Psi$ on state space $\statespace$. 
\end{definition}
 
\begin{Remark}
For any given network $\statespace$ and rate functions $\GRate_i(\psi)$, where $\psi\subset\statespace$, we assert that  proportionally fair, determinantal
MAC schedulers exist because the expectation  $\E[\Thpt(\Psi)]$ is a continuous function of the kernel  $K$ of $\Psi$, and the set of (symmetric, non-negative definite) matrices with bounded eigenvalues is compact. Unfortunately, the problem~\eqref{e.Det-Utility-optimization} is not a convex optimization problem.
\end{Remark}

\subsection{Scheduler as a parametric $L$-ensemble}
\label{ss.Lensemble}
To render the above optimization problem convex, we restrict our attention to a practical class of determinantal kernels. 
\begin{definition} [Determinantal $L$-ensembles]\label{def.Lpp}
For a real, semi-definite matrix~$L$ indexed by the points of a state space $\statespace$, $L$-ensemble  $\Psi$ on state space $\statespace$ is defined as a determinantal point process with a kernel of the form
\begin{equation}\label{e.K-L}
K=L(L+I)^{-1} \,.
\end{equation}
where $K$ is a determinantal kernel, as given in Definition~\ref{def.detpp}. Moreover, for all sets $\psi\subset\statespace$, the $L$-ensemble $\Psi$ admits the probabilities
\begin{equation}\label{e.probL}
\Prob(\,\Psi=\event\,) = \frac{\det(L_{\event}) }{\det{(L+I)}}\,,
\end{equation} 
where $L_{\event}$ denotes the restriction of $L$ to the  points $\psi\subset\statespace$.
\end{definition}
Compare the determinantal probability expression~\eqref{e.probdet} with the $L$-ensemble version given by expression~\eqref{e.probL}, which implies the requirement
 $\det{(L+I)}\neq 0$. The kernel $L$ has the same the eigenvectors as $K$, but its eigenvalues are equal to $\lambda_i/(1+\lambda_i)$, where $\lambda_i$ are the eigenvalues of $L$, so the eigenvalues of the  matrix $L$  need not be smaller than one. Any determinantal kernel $K$ having all eigenvalues strictly smaller than one has the form~\eqref{e.K-L} with $L=(I-K)^{-1}-I$. Restricting to $L$-ensembles is equivalent to assigning a nonzero probability to the empty set; see Kulesza and Taskar~\cite[Section 2.2]{kulesza2012determinantal}. 

These determinantal {\em $L$-ensembles} were originally studied in mathematical physics by Borodin and Rains~\cite{borodin2005eynard} and then later introduced to machine learning by Kulesza and Taskar~\cite{kulesza2012determinantal}, due to the closed-form probability expression~\eqref{e.probL} allowing for model training. Usually researchers assume that the matrix $L$ is symmetric, but there is recent work on non-symmetric versions~\cite{brunel2018learning,gartrell2019learning
}.

We introduce a parametric version of the kernel matrix $L$ studied by Kulesza and Taskar~\cite{kulesza2012determinantal}. For a real, non-negative definite matrix $S$ indexed by the elements of $\statespace$, we define an $L$-ensemble matrix (kernel) $L=L(S,q)$ so that the latter has the matrix elements
\begin{equation}\label{e.qSq}
[L]_{x_i,x_j}=[L(S,q)]_{x_i,x_j}:=q(x_i)Sq(x_j),    
\end{equation}
where $q(\cdot)\ge0$ is some function defined on $\statespace$. The matrix $S$ is called the {\em similarity matrix}, where the value $[S]_{x_i,x_j}$ measures how similar the points $x_i$ and $x_j$ are. In practice, often the general form  $[S]_{x_i,x_j}= \rho(|x_i-x_j|)$ is assumed, where $\rho$ is some non-negative (radial) function, giving a symmetric matrix $S$. The function (vector) $q$ is the \emph{quality}, which places a weight on each point, meaning points with higher qualities are more likely to be randomly selected. We use the term {\em $L(S,q)$-ensemble} to refer to a determinantal point process having a kernel $K$ in the form~\eqref{e.K-L} with the matrix $L=L(S,q)$ given by expression~\eqref{e.qSq}. 

\begin{Example}
If the matrix $S=I$, then the resulting $L(S,q)$-ensemble is simply the independent thinning of $\statespace$ with probabilities~$p(x_i)=q^2(x_i)/(1+q^2(x_i))$, which brings us back to adaptive Aloha.
\end{Example}

When training determinantal models, usually it is the quality function $q$ that is trained, whereas $S$ is  treated as a hyperparameter, motivating another utility function. Recalling the state space cardinality $n=|\statespace|$, for a given similarity matrix $S$, we define a utility function
from the orthant $[0,\infty)^{n}$ to the real line $\R$ as
\begin{equation}
U_S(q):=
\sum_i\log(\E[\GRate_i(\Psi)])=\sum_i\log(\Thpt_i)\,,
\end{equation}
where the random scheduler $\Psi$ is now an $L(S,q)$-ensemble.For a given matrix $S$, we write  $\mathcal{L}_q(\Psi;S)$ to denote all the parametric determinantal  $L(S,q)$-ensemble distributions on the space $\statespace$, which corresponds to all the possible quality functions $q$.
We now define the $L(S,q)$-version of our proportionally fair scheduler  defined in Definition~\ref{def.profair}.
\begin{definition}[Proportionally fair $L(S,q^*)$ scheduler]\label{def.profairL}
For a given $S$, a {\em proportionally fair $L(S,q^*)$ MAC scheduler} is any solution to the optimization problem
 \begin{equation}\label{e.L-Utility-optimization}
\argmax_{\mathcal{L}_q(\Psi;S)}\sum_{i}\log(\Thpt_i)=\argmax_{q\in[0,\infty)^{n}}\sum_i\log(\Thpt_i)\,,
\end{equation}
\end{definition}

When optimizing and fitting any function, we are interested whether it is convex (or concave). This motivates the next result, for which we write  $e^w:=(e^{w(x_i)}:x_i\in\statespace)$ for $w=(w(x_i):x_i\in\statespace)\in[0,\infty)^{n}$. 
\begin{Proposition}
The utility function~$U_S(e^w)$ is a concave function of $w\in [0,\infty)^{n}$, implying any local maximum of  $U_S(e^w)$ identifies a proportionally fair $L(S,e^w)$   scheduler.
\end{Proposition}
\begin{proof}
The proofs follows in a similar fashion to that of Kulesza and Taskar~\cite[Proposition~4.1]{kulesza2012determinantal}.
\end{proof}

\subsection{Solution summary}
For the optimization problem stated in Definition~\ref{def.Psiopt}, the different types of solutions are represented as a Venn diagram in Figure~\ref{fig:solutions}, which shows that all feasible solutions encompass all determinantal $K$  solutions, which encompass all determinantal $L$ solutions, and so on. In this hierarchy, a non-adaptive Aloha solution for a fixed $p$ is a single point inside the set of adaptive Aloha solutions. This solution or model hierarchy is typical when working with statistical, machine learning and artificial intelligence models, where one endeavors under the assumption (or hope) that the optimal model will be close to a suitable and learnable model. Ideally, for our model the optimum or near optimum solutions form part of the $L$-ensemble solutions. 

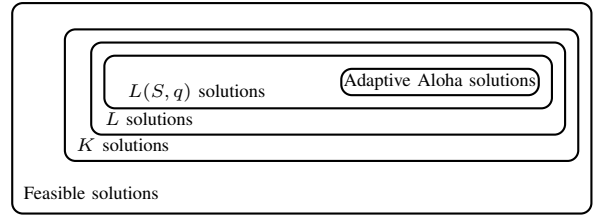
\begin{figure}
    \centering
    \begin{tikzpicture}[thick, scale=0.70]
\draw[rounded corners] (0, 0)  rectangle ++(11, 4) {} (1.5,0.4) node{\scriptsize
 Feasible solutions} ;
\draw[rounded corners] (1, 1) rectangle ++(9.75, 2.5) (2.1,1.3) node{\scriptsize
$K$ solutions} ; 
\draw[rounded corners] (1.5, 1.5) rectangle ++(9, 1.75) (2.6,1.8) node{\scriptsize
$L$ solutions};
\draw[rounded corners] (1.75, 2) rectangle ++(8.5, 1) (3.5,2.3) node{\scriptsize
$L(S,q)$ solutions};
\draw[rounded corners] (6.25, 2.25) rectangle ++(3.75, .5) (8.12,2.5) node{\scriptsize
Adaptive Aloha solutions};
\end{tikzpicture}
\caption{Solution hierarchy for the  problem in Definition~\ref{def.Psiopt}. }
\vspace{-4.5ex}
\end{figure}\label{fig:solutions}

\subsection{Learning the quality model through node features}
\label{ss.Genarla-learning}
The distribution of any proportionally fair $L(S,q^*)$-ensemble
$\Psi$, which is a solution of~\eqref{e.L-Utility-optimization},
depends on the given network  $\statespace$, the (optimal) quality $q^*=q^*(\statespace)$, and, implicitly, the rates $\GRate(\psi)$ for $\psi\subset\statespace$.  In other words, when $\statespace$ is replaced by some
other network $\statespace'$, then one has to
solve the optimization problem~\eqref{e.L-Utility-optimization} again to obtain another
$L(S,q^*(\statespace'))$-ensemble. In the spirit of statistical learning, we ultimately want to  remove this step of optimizing for each new network $\statespace$.

On this note, it is often convenient to further approximate the 
dependence $q^*(\statespace)$  of the optimal  $L$-ensemble on the network $\statespace$ by
some more explicit relation, where the quality vectors $q(x_i)$ of
nodes $x_i\in\statespace$ depend on some easily accessible characteristics  
of the nodes $x_i$ in the network $\statespace$ via some simple
(say log-affine) function that is fitted by considering
some class of networks $\statespace$.

To this end, let $\{f_i(\statespace)\in\R^M: x_i\in\statespace\}$ be some given, real,
$M$-dimensional marks of the nodes $x_i\in\statespace$, called the {\em quality features} of nodes $x_i$.
We can think of them as describing the local geometry of all nodes in $\statespace$ around
$x_i$. 
For each  $x_i$, we assume that the quality vector takes the form   
\begin{equation}\label{e.quality}
q(x_i)=q(x_i;\theta,\statespace):= e^{\theta^\top f(x_i;\statespace) }\,,
\end{equation}
where the vector $\theta\in\R^M$ is the fitting parameter of the quality model and the vector function $f(x_i;\statespace)$ is the feature vector of $x_i$ in the network~$\statespace$, giving 
the scalar product $\theta^\top f(x_i;\statespace)$.
For example, one can take the two-element feature vector
$
f(x_i;\theta,\statespace)=\Bigl(1,\min_{j\not=i}(|x_i-y_j|)\Bigr).
$

\subsection{Training the model}
We just introduced a parametric model for the quality features, which is a stepping stone for statistically training (or fitting) determinantal models. It is tempting to train our model on pre-optimized networks in a similar fashion to previous work~\cite{blaszczyszyn2018determinantal,saha2019machine,tu2023determinantal,tu2025determinantal}. But that is beyond the scope of the current work, where we have instead concentrated on determinantal scheduling with proportional fairness. 

We will just point out that determinantal models can be trained using maximum-likelihood methods, due to the convenient mathematical properties of determinants. The training approach hinges upon the fact that the function $-\log[\det(A)]$, where $\det(A)$ is a determinant of a positive semi-definite matrix $A$, is a convex function. 
For more details, see the textbook~\cite[Chapter 4]{kulesza2012determinantal} and the aforementioned work~\cite{blaszczyszyn2018determinantal,saha2019machine,tu2023determinantal,tu2025determinantal}.

\section{Bi-pole network model  with determinantal MAC scheduling}
\label{s.Network}
We will now apply our determinantal MAC scheduler to a wireless network model with fixed transmitters and receivers, and random propagation effects. 
Consider a fixed (that is, non-random) configuration of potential transmitters on the plane written as $\statespace=\{x_i\}_{i=1}^n\subset\R^2$. For each transmitter $x_i\in \statespace$, we consider  a receiver located at $y_i$, which, in point process terminology, is simply a mark, resulting in a marked point process $\tilde{\statespace}=\{(x_i,y_i)\}_{i=1}^n$. This is a bi-pole model~\cite[Chapter 16]{FnT2} coupled with determinantal scheduling, which we introduced in a previous paper~\cite[Section III]{blaszczyszyn2020coverage}.

We write $P_{x,y}$ to denote the received power at location $y$ of a signal emanating from a transmitter located at $x$. For a transmitter $x_i\in \statespace$, its SINR at location $y\in\R^2$ is given by
\begin{equation}
\SINR(x_i,y)= \frac{ P_{x_i,y}}{W+\sum\limits_{x_j\in \statespace\setminus {x_i}} P_{x_j,y}  } \,,
\end{equation}
where $W$ denotes the noise power.
We can make some standard assumptions on the propagation model. For a signal traversing a distance $|x-y|$, we assume it undergoes path loss according to the  function
$\ell(x-y)=(\kappa|x-y|)^{-\beta }$, where $\kappa>0$ and $\beta>0$. For any point $y\in \R^2 $, let the random variable $F_{x_i,y}$ describe the  fading of the signal propagating from transmitter $x_j$ to location $y$. The resulting SINR expression is
\begin{equation}
\SINR(x_i,y,\statespace)= \frac{ F_{x_i,y} \ell(x_i-y) }{\noise+\sum\limits_{x_j\in \statespace\setminus {x_i}}  F_{x_i,y}\ell(x_j-y)  } \,.
\end{equation}

\subsection{Coverage probabilities under determinantal scheduler}
We use a determinantal kernel $K$ on the set of transmitters $\statespace$, as described in Section~\ref{ss.DetMAC}. We assume that the kernel $K=K(\tilde\statespace)$ depends on the configuration of all potential transmitters and their receivers
$\tilde\statespace=\{(x_i,y_i)\}$, but not on their fading variables.
The discrete point process $\Psi$ is a determinantal thinning of the transmitter configuration $\statespace$.

For a function $f(x_i)$, where $x_i\in \statespace$, we write $K\{f\}$ to refer to a (matrix) kernel defined as
\begin{equation}
    [K\{f\}]_{x_i,x_j}:= \sqrt{1-f(x_i)}[K]_{x_i,x_j}\sqrt{1-f(x_j)}\,. 
\end{equation}
For $z\in\statespace$, we write 
$K^!_{z}$ to denote the kernel indexed by points of $\statespace\setminus\{z\}$ with entries
\begin{equation}
[K_z^!]_{x_i,y_i}=[K]_{x_i,y_i}-
\frac{[K]_{x_i,z}[K]_{y_i,z}}{[K]_{z,z}}\quad x_i,x_j\in\statespace\setminus\{z\}.
\end{equation}
This is the kernel of the reduced Palm version $\Psi^!_z$ of $\Psi$ given the point
$z\in\Psi$, hence the probability
\begin{equation}
\Prob\{\Psi_z^!\cup\{z\}\supseteq\psi\}:=
\Prob\{\Psi\supseteq\psi\,|\,z\in\Psi\}=\det((K^!_z)_\psi)
\end{equation}
for $\psi\subset\statespace$. Similar results exist for conditioning on multiple points; see Borodin and Rains~\cite[Proposition 1.2]{borodin2005eynard}. 
\begin{Example}\label{ex.Palm}
Consider a state space or network $\statespace$ consisting of three points, $\statespace=\{x_1,x_2,x_3\}$, and define a determinantal point process with the kernel $[K]_{i,j}=k_{i,j}$, where $1\leq i,j\leq 3$.
The diagonals are probabilities, meaning $0\leq k_{11},k_{22},k_{33}\leq 1$. For  point $x_1$, the reduced Palm version of the kernel is
\begin{equation}
K_{x_1}^!=    \begin{bmatrix} 
 k_{22}  & k_{23} \\ 
k_{23}  &  k_{33}
\end{bmatrix} -\frac{1}{k_{11}}
\begin{bmatrix} 
 k_{12}^2 & k_{12}k_{13} \\ 
k_{12}k_{13}  &  k_{13}^2
\end{bmatrix}.
\end{equation}
\end{Example}

We now define some functions for the forthcoming result, Proposition~\ref{p.SINR-conditional}, on the successful coverage probability.
For $\kappa$ and $\beta$ parameters of the signal propagation model, as well as the noise power $W$ and SINR threshold  $\tau>0$  (see Section~\ref{s.Network}), we define the functions
\begin{align}\label{e.h1}
h(s,r)&:=1-\frac{1}{\frac{1}{\threshold}(s/r)^{\beta}+1}\quad s,r\ge0,\\
w(s)&:= e^{- (\threshold \noise) (\kappa s)^{\beta}}\quad s\ge0.\label{e.w1}, 
\end{align} 
where  of $s,r>0$. These functions come from a  result that we presented in a previous paper~\cite[Lemma II.1]{blaszczyszyn2020coverage}. 
They are useful for expressing the SINR coverage probability in deterministic network models with Rayleigh fading and power-law path loss function. 
Furthermore, for a  point $x_i\in\statespace$ define the  functions
\begin{equation}
h_{x_i}(x_j):=h(|x_j-y_i|,|x_i-y_i|),\quad x_j\in\statespace\setminus\{x_i\}
\end{equation}
and constants $W_{x_i}:=w(|x_i-y_i|)$.

Now consider determinantal scheduler $\Psi$ with kernel~$K$ on the network configuration $\statespace$. We now present a corresponding result, which we first presented in the paper~\cite[Proposition IV.1.]{blaszczyszyn2020coverage} with the proof~\cite[Appendix A]{blaszczyszyn2020coverage}.
\begin{Proposition}\label{p.SINR-conditional}
For given $x_i\in\statespace$ and $\tau\ge0$, the distribution of the SINR  at $y_i$ for a signal coming from $x_i$, given $x_i$ is selected by the scheduler, 
is obtained with expression
\begin{equation}
\Prob\{\,\SINR(x_i,y_i,\Psi)>\tau\,|\,x_i\in\Psi\,\}=
\det(I-K^!_{x_i}\{h_{x_i}\})W_{x_i},
\end{equation}
where the probability above also incorporates the random fading conditions and scheduler decisions.
\end{Proposition}

We now consider the unconditional SINR (tail) distribution at $y_i$ from $x_i\in\statespace$, namely 
\begin{equation}
\calP_i(\tau):= \Prob\{\,x_i\in\Psi\text{\ and\ \ } \SINR(x_i,y_i,\Psi)>\tau\,\},\quad \tau>0. 
\end{equation}
Then we have an expression for the coverage probability $\calP_i$.
\begin{Corollary}
The coverage probability is given by
\begin{equation}\label{e.covprob}
\calP_i(\tau)=[K]_{x_i,x_i}\det(I-K^!_{x_i}\{h_{x_i}\})W_{x_i}.
\end{equation}
\end{Corollary}
\begin{proof}
Note that $\Prob\{x_i\in\Psi\}=
\det(K_{\{x_i\}})=[K]_{x_i,x_i}$. Then we use Proposition~\ref{p.SINR-conditional} by conditioning on $x_i\in\Psi$.
\end{proof}
The matrix $(I-K^!_{x_i})\{h_{x_i}\}$ on the right-hand-side of~\eqref{e.covprob} is indexed by the elements of $\statespace\setminus\{x_i\}$. We extend it to the full configuration $\statespace$
by setting
\begin{equation}\label{e.Kt}    
[K_{x_i}(\tau)]_{x_j,x_k} := \left\{
	\begin{array}{ll}
		[I - K_{x_i}^!\{h_{x_i}\}]_{x_j,x_k}  & \mbox{if } j,k\not=i \\
				W_{x_i} [K]_{x_i,x_i}  & \mbox{if } j=k=i,\\
		0 &\mbox{if $j=i$ and $k\not=i$}\\
             0 &\mbox{if $k=i$ and $j\not=i$}\,  
			\end{array}
	\right. 
\end{equation}
This allows us to express $\calP_i(\tau)$ in a more compact way
\begin{equation}\label{e.covprob-det}
\calP_i(\tau)=\det(K_{x_i}(\tau)).     
\end{equation}

Expression~\eqref{e.covprob} for the coverage probability $\calP_i$ is straightforward to implement for numerical calculations; see the repository~\cite{keeler2025detschedule} for the file \texttt{CovProbDet.m}.

\subsection{Transmission rates and throughput}
For our bi-pole model, we are interested in quantifying the communication rate of a wireless link between a transmitter $x_i$ and its receiver $y_i$. 
For a transmitter-and-receiver pair $(x_i,y_i)$, we introduce the \emph{constant rate} model as
\begin{equation}\label{e.Rate-constant}
\GRate_i:=\Rate_{\tau,R_0}(x_i,y_i,\statespace):= R_0  \Prob[\,\SINR(x_i,y_i,\statespace)> \threshold\,] \, ,
\end{equation}
where $R_0>0$ is a constant 
and $\threshold$ is a fixed SINR threshold;
the probability corresponds to random fading variables.
We introduce the \emph{variable rate} model as
\begin{equation}\label{e.Rate-variable}
\Rate_r(x_i,y_i,\statespace):= \E[r(\SINR(x_i,y_i,\statespace))] \, ,
\end{equation}
where $r$ is some non-negative function that maps from the non-negative  real line, that is, $r:\R^+\rightarrow \R^+ $. The above expectation corresponds to random fading variables.
Based on Shannon's work, a natural choice is $r(t)=C\log(t+1)$, where the constant $C>0$.  

The explicit expressions~\eqref{e.covprob} and~\eqref{e.covprob-det} for the coverage probability~$\calP_i(\tau)$ immediately imply the following representation of the user throughput defined by the general expression~\eqref{e.Throughput}. 
\begin{Corollary}
The throughput $\Thpt_i$ of $x_i$  is given by 
\begin{equation}
\Thpt_i=R_0 \calP_i(\tau)\,,
\end{equation}
in the case of constant rate function~\eqref{e.Rate-constant}
and by 
\begin{equation}
\Thpt_i=\int_0^\infty \calP_i(r^{-1}(\tau))\,d\tau
\end{equation} 
in the case of variable rate function~\eqref{e.Rate-variable}, provided the inverse $r^{-1}$ 
of the rate function~$r$ exists.
\end{Corollary}

\begin{Remark}
For the time-averaged transmission rate $\Thpt_i=\E[R_i(\Psi)]$, the constant rate function~\eqref{e.Rate-constant}
gives the  expression 
\begin{align}
\sum_{x_i\in\statespace}\log(\Thpt_i)&= 
m\log R_0+\sum_{i=1}^m \sum_{j=1}^m \log\lambda_{x_i,j}(\tau)
\end{align}
where $\lambda_{x_i,j}(\tau)$ is $j$-th eigenvalue of the matrix $K_{x_i}(\tau)$ defined in equation~\eqref{e.Kt}. 
(Recall that a determinant of a square matrix is equal to the product of its eigenvalues.)
We do not know whether this expression can 
help solve the problem of the proportionally fair determinantal scheduler~\eqref{e.Det-Utility-optimization}. 
\end{Remark}

\subsection{Controlling medium access and separating transmissions}
We recall that an $L$-ensemble $\Psi$ with matrix $L$ satisfies $\Prob\{\,\Psi=\psi\,\}={\det L_{\psi}}/{\det(L+I)}$. Consequently, the probability of the configuration  $\psi=\{x_{i_1},\ldots,x_{i_k}\}\in\statespace$ being scheduled for the transmission is equal to 
\begin{equation}\label{e.quality-diversity}
\Prob\{\,\psi=\Psi\,\}=
\frac{1}{\det(L+I)}
\prod_{x_{i}\in\psi} q^2(x_i)\det(S_\psi).
\end{equation}
In expression~\eqref{e.quality-diversity}, the quality-squared term
$q^2(x_i)$ controls the medium access probability of each node $x_i$, whereas the matrix $S$ controls the separation of 
different transmissions. The term 
$\det(L+I)$ is just the normalizing constant.

For the matrix $S$, we can use, 
for example, the {\em Gaussian kernel} 
\begin{equation}\label{eq.SGauss}
    [S]_{x_i,x_j}:= e^{-|x_i-x_j|^2/\sigma^2},
\end{equation}
where the scale parameter $\sigma$ controls separation of transmitting nodes.  
As $\sigma\rightarrow0$, the matrix $S$ converges to the identity matrix, giving the (adaptive) Aloha model. 


\section{Numerical example}
Using our fairness framework, we present a numerical example of the bi-pole network model outlined in Section~\ref{s.Network}; see Figure~\ref{fig:CovProbBipolarA} and Figure~\ref{fig:CovProbBipolarB}. We randomly placed a fixed number of transmitters in a square unit window, forming a binomial point process. For each transmitter, its receiver is located uniformly inside a disc with radius $r_{\max}$ centred at the transmitter; see the code  online~\cite{keeler2025detschedule}.

We compared the determinantal model to the classic Aloha model (with fixed access probability) and the adaptive Aloha model (with varying access probability).  
For each network layout, we maximized its total logarithmic utility (or fairness) with the rate set as the  coverage probability $\calP_i$ given by equation~\eqref{e.covprob-det}. For the two Aloha models, the similarity matrix $S$ is simply the $n\times n$ identity matrix, where $n$ is the number of transmitter-and-receiver pairs. For the determinantal model, we opted for the simple Gaussian $S$ kernel~\eqref{eq.SGauss}. For the determinantal and adaptive Aloha models, the optimization happens over $n$ variables, while the fixed Aloha model only has one variable for optimizing. We used the parameters: SINR threshold $\tau=10$, path loss model $\ell(r)=(1+r)^{-\beta}$, path loss exponent $\beta=4$, exponential fading mean $\mu$, maximum transmission radius $r_{\max}=0.1$, and  Gaussian kernel parameter $\sigma=10$. The parameters were only chosen to illustrate the models in use, not for model accuracy. 

We observed that often there was not a significant difference between adaptive Aloha and fixed Aloha. We only looked at small networks with five and 10 transmitter-receiver pairs. Naturally, as there are more pairs, the coverage probability decreases due to the increase in interference in the network.  We stress that much more numerical investigation remains to be carried out. 

\begin{figure}[t]
\begin{minipage}[b]{0.48\linewidth}
\centering
\centerline{\includegraphics[width=1.1\linewidth]{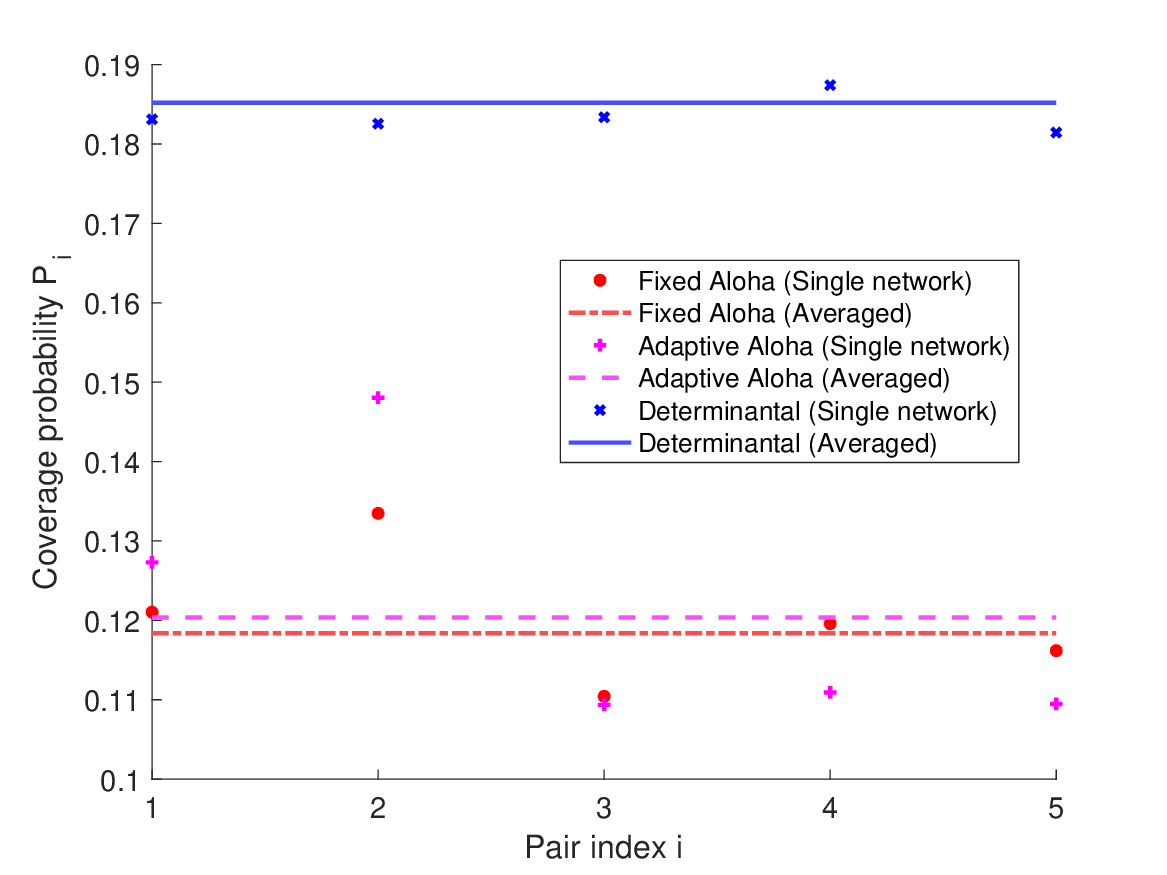}}
\vspace{-2ex}
\caption{\footnotesize Coverage probabilities (single realization and averaged) for a five-pair bi-pole network.}
\label{fig:CovProbBipolarA}
\end{minipage}
\hspace{0.1em}
\begin{minipage}[b]{0.48\linewidth}
\centering
\centerline{\includegraphics[width=1.1\linewidth]{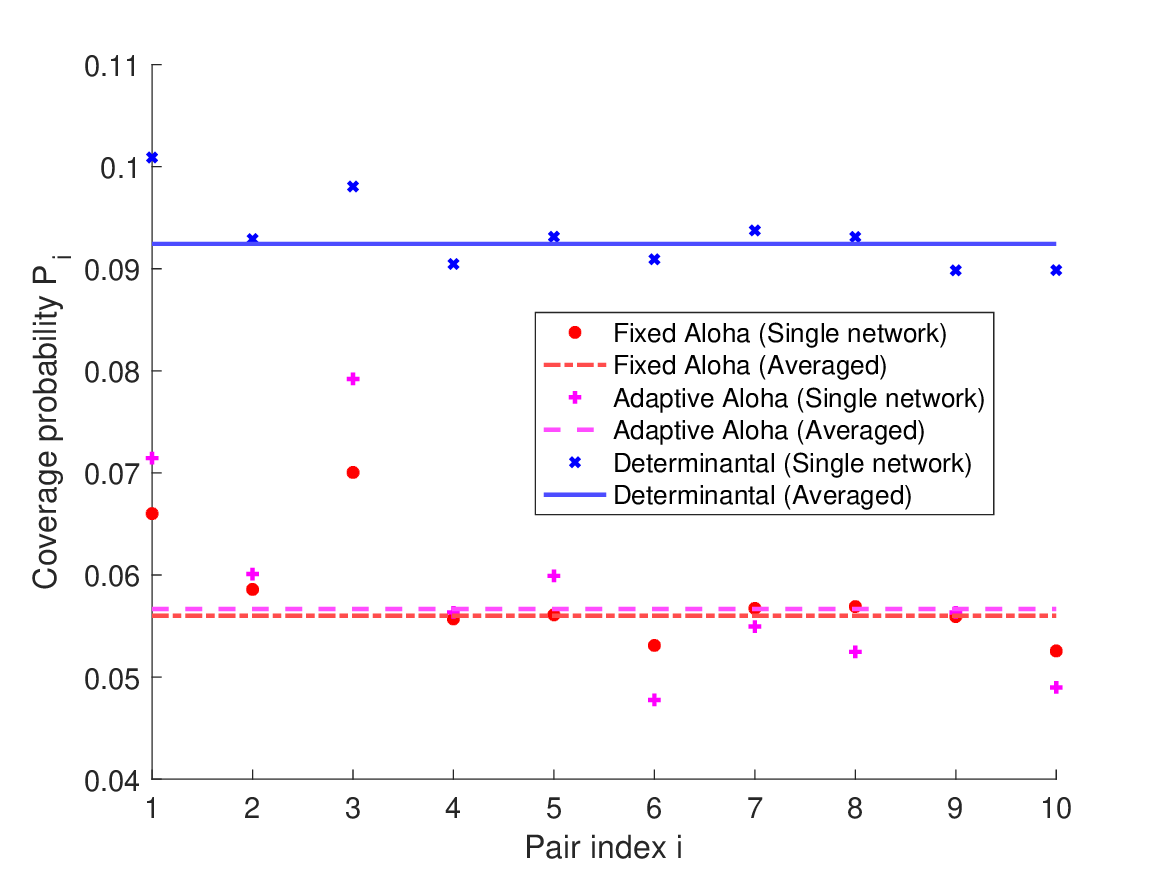}}
\vspace{-2ex}
\caption{\footnotesize Coverage probabilities (single realization and averaged) for a ten-pair bi-pole network.}
\label{fig:CovProbBipolarB}
\end{minipage}
\vspace{-3ex}
\end{figure}

\section{Conclusion}
We have presented an adaptive determinantal scheduling framework for wireless networks. By formulating the scheduling problem within the context of proportional fairness and leveraging $L$-ensembles, we have demonstrated that determinantal processes offer a mathematically elegant and computationally tractable approach for network scheduling.
Furthermore, the proposed approach is adaptable to various network configurations without requiring repeated optimization for each new network deployment.

Promising directions for future research emerge from this work.
\textbf{Training on pre-optimized networks}: Using pre-optimized networks to trains models with determinantal fairness framework  is the first obvious step. 
 \textbf{Multi-objective optimization:} Incorporating additional performance metrics beyond proportional fairness, such as energy efficiency, delay constraints, or quality-of-service requirements, could lead to more comprehensive scheduling solutions that balance multiple competing objectives.
\textbf{Distributed scheduling:} Developing decentralized versions of determinantal scheduling algorithms would aid their deployment in large-scale networks where centralized control is impractical. 
\textbf{Leveraging machine learning:} Further exploration of the connection between determinantal processes and machine learning techniques could yield powerful approaches. 
\textbf{Theoretical performance bounds:} Establishing theoretical guarantees on the performance gap between optimal scheduling and determinantal approximations would provide valuable insights into the fundamental limits of determinantal scheduling approaches.

In conclusion, determinantal point processes offer a promising mathematical foundation for wireless network scheduling that balances computational tractability with sophisticated spatial modelling. 

\addtocounter{section}{1}
\addcontentsline{toc}{section}{References}
{\scriptsize
\bibliography{Det}}
\bibliographystyle{abbrv}

\end{document}